\documentclass[a4paper]{article}

\usepackage{amsmath}
\usepackage{amssymb}
\usepackage{amsfonts}
\usepackage{graphicx}
\usepackage{algorithm}
\usepackage{url ,}
\usepackage[usenames,dvipsnames]{color}

\DeclareMathOperator{\E}{\mathbb{E}}

\DeclareMathOperator{\poly}{poly}

\DeclareMathOperator*{\argmin}{arg\,min}

\newtheorem{theorem}{Theorem}[section]
\newtheorem{lemma}[theorem]{Lemma}

\newtheorem{definition}{Definition}[section]

\let\oldmarginpar\marginpar
\renewcommand\marginpar[1]{\-\oldmarginpar[\raggedleft\footnotesize #1]%
{\raggedright\footnotesize #1}}

\usepackage{comment}
\usepackage{color} 
\usepackage{algorithm} 
\usepackage{algorithmic} 

\AtBeginDocument{%
 \abovedisplayskip=2pt plus 1pt minus 1pt
 \abovedisplayshortskip=0pt plus 1pt
 \belowdisplayskip=2pt plus 1pt minus 1pt
 \belowdisplayshortskip=0pt plus 1pt
}
\usepackage{array}
\usepackage{amsmath} 
\usepackage{amssymb} 
\usepackage{graphicx} 
\usepackage{enumitem} 
\usepackage{float}
\usepackage{pdfsync}
\usepackage{verbatim}

\usepackage[pagebackref=true,colorlinks]{hyperref}
\hypersetup{linkcolor=red,filecolor=red,citecolor=red,urlcolor=red}

\newcommand{\idlow}[1]{\mathord{\mathcode`\-="702D\it #1\mathcode`\-="2200}}
\newcommand{\id}[1]{\ensuremath{\idlow{#1}}}
\newcommand{\litlow}[1]{\mathord{\mathcode`\-="702D\sf #1\mathcode`\-="2200}}
\newcommand{\lit}[1]{\ensuremath{\litlow{#1}}}

\newcommand{\namedref}[2]{\hyperref[#2]{#1~\ref*{#2}}}

\newcommand{\figurerefb}[2]{\hyperref[#1]{Figure~\ref*{#1}#2}}

\newcommand{\equationref}[1]{\hyperref[#1]{(\ref*{#1})}}
\renewcommand{\eqref}{\equationref}

\usepackage[textsize=tiny]{todonotes}

\newcommand{\DEBUG}[1]{}

\usepackage{subfigure}

\renewcommand{\paragraph}[1]{\noindent \textbf{#1}}

\newenvironment{proof}{\textit{Proof}.}{\hfill$\square$}

\date{}
\title{Distributionally Linearizable Data Structures}
\author{Dan Alistarh \\ IST Austria \and Trevor Brown \\ IST Austria \and Justin Kopinsky \\ MIT \and Jerry Z. Li \\MIT \and Giorgi Nadiradze \\ ETH Zurich}

\begin{document}
\maketitle
\begin{abstract}
Relaxed concurrent data structures have become increasingly popular, due to their scalability in graph processing and machine learning applications (\cite{Nguyen13, gonzalez2012powergraph}). 
%
%
%
%
Despite considerable interest, there exist families of natural, high performing \emph{randomized} relaxed concurrent data structures, such as the popular MultiQueue~\cite{MQ} pattern for implementing relaxed priority queue data structures,  
for which no guarantees are known in the concurrent setting~\cite{AKLN17}. 

Our main contribution is in showing for the first time that, under a set of analytic assumptions, a family of relaxed concurrent data structures, including variants of MultiQueues, but also a new approximate counting algorithm we call the MultiCounter, 
 provides strong probabilistic guarantees on the degree of relaxation with respect to the sequential specification, \emph{in arbitrary concurrent executions}. 
We formalize these guarantees via a new correctness condition called \emph{distributional linearizability}, tailored to concurrent implementations with randomized relaxations. 
Our result is based on a new analysis of an \emph{asynchronous} variant of the classic power-of-two-choices load balancing algorithm, in which placement choices can be based on inconsistent, outdated information (this result may be of independent interest). 
We validate our results empirically, showing that the MultiCounter algorithm can implement scalable relaxed timestamps, which in turn can improve the performance of the classic TL2 transactional algorithm by  up to $3\times$, for some settings of parameters.

%
%

\end{abstract}

\section{Introduction}

Consider a system of $n$ threads, which share a set of $n$ distinct atomic counters. We wish to implement a \emph{scalable approximate counter}, which we will call a \emph{MultiCounter}, by distributing the contention among these $n$ distinct instances: 
to \emph{increment} the global counter, a thread selects two atomic counters $i$ and $j$ uniformly at random, reads their values, and (atomically) increments by $1$ the value of the one which has \emph{lower value} according to the values it read.   
To \emph{read} the global counter, the thread returns the value of a randomly chosen counter $i$, multiplied by $n$. \footnote{This multiplication serves to maintain the same magnitude as the total number of updates to the distributed counter up to a point in time.}

The astute reader will have noticed that this process is similar to the classic two-choice load balancing process~\cite{ABKU}, in which a sequence of balls are placed into $n$ initially empty bins, 
and, in each step, a new ball is placed into the less loaded of two randomly chosen bins. Here, the individual atomic counters are the \emph{bins}, and each increment corresponds to a new \emph{ball} being added. 
This sequential load balancing process is extremely well studied~\cite{Richa01, Mitz}: a series of deep technical results established that the difference between the most loaded bin and the average is $O( \log \log n )$ in expectation~\cite{ABKU, Mitz}, and that this difference remains stable as the process executes for increasingly many steps~\cite{Berenbrink00, PTW15}. 
We would therefore expect the above relaxed concurrent counter to have relatively low and stable skew among the outputs at consecutive operations, and to scale well, as contention is distributed among the $n$ counters.

However, there are several technical issues when attempting to analyze this natural process in a concurrent setting. 
\begin{itemize} 

\item First, concurrency interacts with classic two-choice load balancing process in non-trivial ways. 
The key property of the two-choice process which ensures good load balancing is that trials are \emph{biased towards less loaded bins}---equivalently, operations are biased towards incrementing counters of lesser value.   
However, this property may break due to concurrency: \emph{at the time of the update}, a thread may end up updating the counter of \emph{higher} value among its two choices if the counter of smaller value is updated concurrently since it was read by the thread, thus surpassing the other counter. 

\item Second, perhaps suprisingly, it is currently unclear how to even \emph{specify} such a  concurrent data structure. Despite a significant amount of work on specifying \emph{deterministic relaxed} data structures~\cite{Henzinger, Quasi, LocalLin} , none of the existing frameworks cover relaxed \emph{randomized} data structures.   

\item Finally, assuming such a data structure can be analyzed and specified, it is not clear whether it would be in any way \emph{useful}: many existing applications are built around data structures with deterministic guarantees, and it is not obvious how scalable, relaxed data structures can be leveraged in standard concurrent settings. 

\end{itemize}

One may find it surprising that analysing such a relatively simple concurrent process is so challenging. 
Beyond this specific instance, these difficulties reflect wider issues in this area: 
although these constructs are reasonably popular in practice due to their good scalability, e.g.~\cite{Basin11, Nguyen13, klsm, MQ}, their properties are non-trivial to pin down~\cite{AKLN17}, 
and it is as of yet unclear how they interact with the higher-order algorithmic applications they are part of~\cite{lenharth2015priority}.

\paragraph{Contribution.} 
In this work, we take a step towards addressing these challenges. Specifically: 
\begin{itemize}
\item We provide the first analysis of a two-choice load balancing process in an asynchronous setting, where operations may be interleaved, and the interleaving is decided by an adversary. 
We show that the resulting process is robust to concurrency, and continues to provide strong balancing guarantees in potentially infinite executions, as long as the ratio between the number of bins and the number of threads is above a large constant threshold. 

\item We introduce a new correctness condition for \emph{randomized relaxed} data structures, called \emph{distributional linearizability}. 
Intuitively, a concurrent data structure $D$ is distributionally linearizable to a \emph{sequential random process} $R$, defined in terms of a sequential specification $S$, a cost function $\id{cost}$ measuring the deviation from the sequential specification, and a distribution $\mathcal{P}$ on the values of the cost function, if every execution of $D$ can be mapped onto an execution of the relaxed sequential process $R$, respecting the outputs and the  costs incurred, as well as the order of non-overlapping operations.  

\item We prove that the randomized \emph{MultiCounter} data structure introduced above is distributionally linearizable to a (sequential) variant of the classic two-choice load balancing process. This allows us to formally define the properties of MultiCounters. 
Moreover, we show that this analytic framework also covers variants of \emph{MultiQueues}~\cite{MQ}, a popular family of concurrent data structures implementing relaxed concurrent priority queues. 
This yields the first analytical guarantees for MultiQueues in concurrent executions. 

\item We implement the MultiCounters, and show that they can provide a highly scalable approximate timestamping mechanism, with relatively low skew. 
We build on this, and show that MultiCounters can be successfully applied to timestamp-based concurrency control mechanisms such as the TL2 software transactional memory protocol~\cite{tl2}. 
This usage scenario presents an unexpected trade-off: assuming low contention, the resulting TM protocol scales almost linearly, but may break correctness with very low probability. 
In particular, we show that there exist workloads and parameter settings for which this relaxed TM protocol scales almost linearly, improving the performance of the TL2 baseline by more than $3\times$, without breaking correctness. 

\end{itemize}

\paragraph{Techniques.} 
Our main technical contribution is the concurrent analysis of the classic two-choice load balancing process, in an asynchronous setting, where the interleaving of low-level steps is decided by an oblivious adversary. 
The core of our analysis builds on the elegant potential method of Peres, Talwar and Wieder~\cite{PTW15}, which we render robust to asynchronous updates based on potentially stale information. 
To achieve this, we overcome two key technical challenges. 
The first is that, given an operation $\id{op}$, as more and more other operations execute between the point where it reads and the point where it updates, the more stale its information becomes, and so 
the probability that $\id{op}$ makes the ``right" choice at the time of update, inserting into the less loaded of its two random choices, \emph{decreases}. Moreover, operations updating with stale information will ``stampede"  towards lower-weight bins, effectively skewing the distribution. 
The second technical issue we overcome is that long-running operations, which experience a lot of concurrency, may in fact be adversarially biased towards the wrong choice, inserting into the more loaded of its two choices with non-trivial probability. We discuss these issues in detail in Section~\ref{sec:prelim}. 

In brief, our analysis circumvents this issues by showing that a variant of the two-choice process where up to a constant fraction of updates are corrupted, in the sense that they perform the ``wrong" update, will still have similar balance properties as the original process. It is interesting to note that even the \emph{order} in which corrupted updates occur can be controlled by the adversary through increased concurrency, which is not the case in standard analyses~\cite{PTW15}. The critical property which we leverage in our analysis is that, while individual operations can be arbitrarily contended (and therefore biased), there is a bound of $n$ on the \emph{average} contention per operation, which in turn bounds the average amount of bias the adversary can induce over a period of time. Our argument formalizes this intuition, and phrases it in terms of the evolution of the potential function. 

We show that this result has implications beyond ``parallelizing" the classic two-choice process, as we can leverage it to obtain probabilistic bounds on the skew of the MultiCounter. 
Using the framework of~\cite{AKLN17}, which connected two-choice load balancing with MultiQueue data structures in the sequential case, we can obtain guarantees for this popular data structure pattern in concurrent executions. 

\section{Related Work}

\paragraph{Randomized Load Balancing.} The classic two-choice balanced allocation process was introduced in~\cite{ABKU}, 
where the authors show that, under two-choice insertion, the most loaded among $n$ bins is at most $O( \log \log n )$ above the average, both in expectation and with high probability.  
The literature studying analyses and extensions of this process is extremely vast, hence we direct the reader to~\cite{Richa01, Mitz} for in-depth surveys of these techniques. 
Considerable effort has been dedicated to understanding guarantees in the ``heavily-loaded" case, where the number of insertion steps is unbounded~\cite{Berenbrink00, PTW15}, 
and in the ``weighted" case, in which ball weights come from a probability distribution~\cite{TW07, Berenbrink08}. 
A tour-de-force by Peres, Talwar, and Wieder~\cite{PTW15} gave a potential argument characterizing a general form of the heavily-loaded, weighted process on \emph{graphs}. Our analysis starts from their framework, and modifies it to analyze a concurrent, adversarial process. 
One significant change from their analysis is that, due to the adversary, changes in the potential are only partly stochastic: most steps might be slightly biased away from the better of the two choices, while a subset of choices might be almost deterministically biased towards the \emph{wrong} choice. Further, the adversary can decide the \emph{order} in which these different steps, with different biases, occur. 

Lenzen and Wattenhofer~\cite{Lenzen} analyzed \emph{parallel} balls-into-bins processes, in which $n$ balls need to be distributed among $n$ bins, under a communication model between the balls and the bins, showing that almost-perfect allocation can be achieved in $O(\log^* n)$ rounds of communication. This setting is quite different from the one we consider here. 
Similar delayed information models, where outdated information is given to the insertion process were considered by Mitzenmacher~\cite{Mitzenmacher00} and by Berenbrink, Czumaj, Englert, Fridetzky, and Nagel~\cite{BCEFN12}. The former reference proposes a bulletin board model with periodic updates, in which information about the load of the model is updated only periodically (every $T$ seconds), and various allocation mechanisms. The author provides an analysis of this process in the asymptotic case (as $n \rightarrow \infty$), supported by simulations. 
The latter reference~\cite{BCEFN12} considers a similar model where balls arrive in \emph{batches}, and must perform allocations collectively based solely on the information available at the beginning of the batch, without additional communication. 
The authors prove that the greedy multiple-choice process preserves its strong load balancing properties in this setting: in particular, the gap between min and max remains $O( \log n )$. 
The key difference between these models and the one we consider is that our model is completely asynchronous, and in fact the interleavings are chosen adversarially. The technique we employ is fundamentally different from those of~\cite{Mitzenmacher00, BCEFN12}. In particular, we believe our techniques could be adapted to re-derive the main result of~\cite{BCEFN12}, albeit with worse constants.

Recent work by a subset of the authors~\cite{AKLN17} analyzed the following producer-consumer process: a set of balls labelled $1, 2, \ldots, b$ are inserted sequentially at random into $n$ bins; in parallel, balls are removed from the bins by always picking the lower labelled (higher priority) of two uniform random choices.\footnote{Balls in each bin are sorted in increasing order of label, i.e. each bin corresponds to a sequential priority queue.} This process sequentially models a series of popular implementations of concurrent priority queue data structures, e.g.~\cite{MQ, Haas}. This process provides the following guarantees: in each step $t$, the expected \emph{rank} of the label removed among labels still present in the system is $O( n )$, and $O( n \log n )$ with high probability in $n$. 
That is, this sequential process provides a structured probabilistic relaxation of a standard priority queue. 

\paragraph{Relaxed Data Structures.} The process considered in~\cite{AKLN17} is sequential, whereas the data structures implemented are concurrent. 
Thus, there was a significant gap between the theoretical guarantees and the practical implementation. Our current work extends to concurrent data structures, closing this gap. 
Under the oblivious adversary assumption and given our parametrization, we  show for the first time that practical data structures such as~\cite{MQ, Haas, AKLN17} provide guarantees in real executions. 

Designing efficient concurrent/parallel data structures with relaxed semantics was initiated by Karp and Zhang~\cite{KarZha93}, with other significant early work by Deo and Prasad~\cite{DeoPra92} and Sanders~\cite{San98}. 
It has recently become an extremely active research area, see e.g.~\cite{LotanShavit, Basin11, klsm, SprayList, Haas, Nguyen13, MQ, AKLN17} for recent examples. 
 To the best of our knowledge, ours is the first analysis of randomized relaxed concurrent data structures which works under arbitrary oblivious schedulers: previous analyses such as~\cite{SprayList, MQ, AKLN17} required strong assumptions on the set of allowable interleavings. 
Dice et al.~\cite{DiceLM13} considered randomized data structures for scalable exact and approximate counting. They consider the efficient parallelization of sequential approximate counting methods, and therefore have a significantly different focus than our work.

\section{System Model}

\paragraph{Asynchronous Shared Memory.} We consider a standard asynchronous shared-memory model, e.g.~\cite{AttiyaWelch},
in which $n$ threads (or processes) $P_1, \ldots, P_{n}$, communicate through shared memory, on which they perform atomic operations such as \lit{read}, \lit{write},  \lit{compare-and-swap} and \lit{fetch-and-increment}. 
The \lit{fetch-and-increment} operation takes no arguments, and returns the value of the register before the increment was performed, incrementing its value by $1$. 

\paragraph{The Oblivious Adversarial Scheduler.} Threads follow an algorithm,  composed of
shared-memory steps and local computation, including random coin flips. 
The order of process steps is controlled by an adversarial entity we call the \emph{scheduler}. 
Time $t$ is measured in terms of the number of shared-memory steps scheduled by the adversary. 
The adversary may choose to crash a set of at most $n - 1$ processes by not scheduling them for the rest of the execution.  
A process that is not crashed at a certain step is \emph{correct}, and if it never crashes then it takes an infinite number of steps in the execution.
In the following, we assume a standard \emph{oblivious} adversarial scheduler, which decides on the interleaving of thread steps independently of the coin flips they produce during the execution.  

\paragraph{Shared Objects.} The algorithms we consider are implementations of shared objects. 
A shared object $O$ is an abstraction providing a set of \emph{methods},
each given by a sequential specification. 
In particular, an implementation of a method $n$ for an object $O$ is a set of $n$ algorithms,
one for each executing process.
When thread $P_i$ invokes method $n$ of object $O$,
it follows the corresponding algorithm until it receives a response from the algorithm.
Upon receiving the response, the process is immediately assigned another method invocation.
In the following, we do not distinguish between a method $n$ and its implementation.
A method invocation is \emph{pending} at some point in the execution if has been initiated but has not yet received a response.
A pending method invocation is \emph{active} if it is made by a \emph{correct} process (note that the process may still crash in the future).
For example, a concurrent counter could implement $\lit{read}$ and $\lit{increment}$ methods, with the same semantics as those of the  sequential data structure. 

\paragraph{Linearizability.} The standard correctness condition for concurrent implementations is \emph{linearizability}~\cite{HerlihyWing}: roughly, a linearizable implementation ensures that each concurrent operation can be seen as executing at a single instant in time, called its linearization point. 
The mapping from method calls to linearization points induces a global order on the method calls, which is guaranteed to be consistent to a sequential execution in terms of the method outputs; moreover, each linearization point must occur between the start and end time of the corresponding method. 

Recent work, e.g.~\cite{Henzinger}, considers deterministic relaxed variants of linearizability, in which operations are allowed to deviate from the sequential specification by a \emph{relaxation} factor. Such relaxations appear to be necessary in the case of data structures such as exact counters or priority queues in order to circumvent strong linear lower bounds on their concurrent complexity~\cite{Alistarh14}. 
While specifying such data structures in the concurrent case is well-studied~\cite{Henzinger, Quasi, LocalLin}, less is known about how to specify structured randomized relaxations. 

\paragraph{With High Probability.} We say that an event occurs \emph{with high probability} in a parameter, e.g. $n$, if it occurs with probability at least $1 - 1 / m^{con}$, for some constant $c \geq 1$.


\section{The MultiCounter Algorithm}

\paragraph{Description.} The algorithm implements an approximate counter by distributing updates among $n$ distinct counters, each of which supports atomic  $\lit{read}$ and $\lit{increment}$ operations. Please see Algorithm~\ref{algo:approx} for pseudocode. To read the counter value, a thread simply picks one of the $n$ counters uniformly at random, and returns its value multiplied by $n$. 
To increment the counter value, the thread picks two counter indices $i$ and $j$ uniformly at random, and reads their current values sequentially. 
It then proceeds to update (increment) the value of the counter which appeared to have a lower value given its two reads. (In case of a tie, or when the two choices are identical, the tie is broken arbitrarily.)

\begin{algorithm}[ht]
\caption{Pseudocode for the MultiCounter Algorithm.}
 \label{algo:approx}
\begin{algorithmic}
\STATE \textbf{Shared}: $\id{Counters}[m]$ \textit{// Array of integers representing set of $n$ distinct counters}

\STATE \textbf{function} \lit{Read}()
	\STATE $i \gets \lit{random}(1, n)$
	\STATE \textbf{return} $n \cdot \id{Counters}[i].\lit{read}()$
\STATE
\STATE \textbf{function} \lit{Increment}()
%
	\STATE $i \gets \lit{random}(1, n)$
	\STATE $j \gets \lit{random}(1, n)$
	\STATE $x_i \gets \id{Counters}[i].\lit{read}()$
	\STATE $x_j \gets \id{Counters}[j].\lit{read}()$
	\STATE $Counters[\argmin(x_i,x_j)].\lit{increment}()$

\end{algorithmic}

\end{algorithm}

\paragraph{Relation to Load Balancing.} A \emph{sequential} version of the above process, in which the counter is read or incremented \emph{atomically}, is identical to the classic two-choice balanced allocation process~\cite{ABKU}, where each counter corresponds to a bin, and each increment corresponds to a new ball being inserted into the less loaded of two randomly chosen bins. 

In a concurrent setting, the critical departure from the sequential model is that the values read can be \emph{inconsistent} with respect to a sequential execution: there may be no single point in time when the two counters had the values $x_i$ and $x_j$ observed by the thread; moreover, these values may change between the point where they are read, and the point where the update is performed. 

More technically, the sequential variant of the two-choice process has the crucial property that, at each increment step, it is ``biased" towards incrementing the counter of lower value. 
This does not necessarily hold for the concurrent approximate counter: for an operation where a large number of updates occur between the read and the update points, the read information is stale, and therefore the thread's increment choice may be no better than a perfectly random one; in fact, as we shall see in the analysis, it is actually possible for an adversary to engineer cases where the algorithm's choice is biased towards incrementing the counter of  \emph{higher value}.

\section{Distributional Linearizability}

We generalize the classic linearizability correctness condition to cover \emph{randomized relaxed} concurrent data structures, such as the MultiCounter. 
Intuitively, we will say that a concurrent data structure $D$ is distributionally linearizable to a corresponding \emph{relaxed sequential process} $R$, defined in terms of a sequential specification $S$, a cost function $\id{cost}$ measuring the deviation from the sequential specification, and a distribution $\mathcal{P}$ on the \emph{cost} function values, such that every execution of $D$ can be mapped onto an execution of the relaxed sequential process $R$, respecting the outputs and the incurred costs, as well as the order of non-overlapping operations.  
To formalize this definition, we introduce the following machinery, part of which is adopted from~\cite{Henzinger}. 

\paragraph{Data Structures and Labeled Transition Systems.} 
Let $\Sigma$ be a set of methods including input and output values. 
A sequential history $s$ is a sequence over $\Sigma$, i.e. an element in $\Sigma^*$. 
A (sequential) data structure is a sequential specification $S$ which is a prefix-closed set of sequential histories. 
For example, the sequential specification of a stack consists of all valid sequences for a stack, i.e. in which every \lit{push} places elements on top of the stack, and every \lit{pop} removes elements from the top of the stack. 

Given a sequential specification $S$, two sequential histories $s, t \in S$ are equivalent, written $s \simeq t$, if they correspond to the same ``state:" formally, for any sequence $u \in \Sigma^*$, $su \in S$ iff $tu \in S$. Let $[s]_S$ be the equivalence class of $s \in S$. 

\begin{definition}
	\label{def:lts}
	Let $S$ be a sequential specification. 
	Its corresponding labeled transition sequence (LTS) is an object $LTS (S) = (Q, \Sigma, \rightarrow, q_0)$, with states $Q = \{ [s]_S | s \in S \}$, set of labels $\Sigma$, transition relation $\rightarrow \subseteq Q \times \Sigma \times Q$ given by 
		$ [s]_S \rightarrow^m [sm]_S$ iff $sm \in S, \textnormal{ and }$
 initial state $q_0 = [\epsilon]_S$.  
\end{definition}

Notice that the sequential specification $S$ can be alternatively defined as the set of all traces of the initial state of $LTS(S)$: formally, for any $u \in \Sigma^*$, we have $u \in S$ iff $q_0 \rightarrow^u$.  

\paragraph{Randomized Quantitative Relaxations.} 
Let $S \in \Sigma^*$ be a data structure with $LTS(S) = (Q, \Sigma, \rightarrow, q_0)$. 
To obtain a randomized quantitative relaxation of $S$, we apply the following four steps. The first three steps are identical to deterministic quantitative relaxations~\cite{Henzinger}, whereas the fourth defines the probability distribution on costs:
\begin{enumerate}
	\item \textbf{Completion}: We start from $LTS(S)$, and construct a completed labeled transition system, with transitions from any state to any other state by any method:
	$$ LTS_c(S) = (Q, \Sigma, Q\times\Sigma\times Q, q_0).$$
	\item \textbf{Cost function:} We add a cost function $cost : Q\times\Sigma\times Q \rightarrow \mathbb{R}$ to the LTS. 
	The transition cost will satisfy
	$$ cost(q, m, q')  = 0 \textnormal{ if and only if } q \rightarrow^{m} q' \textnormal{ in } LTS(S).$$
	A quantitative path is a sequence 
	$$ \kappa = q_1 \rightarrow^{m_1, k_1} q_2 \rightarrow^{m_2, k_2} \ldots \rightarrow^{m_n, k_n} q_{n + 1}. $$
	We call the  sequence $\tau = (m_1, k_1), \ldots, (m_n, k_n)$ of transitions and costs the quantitative trace of $\kappa$, denoted by $qtr(\kappa)$. 
	
	\item \textbf{Path cost function}: Given a quantitative path $\kappa$, its path cost is defined as $pcost : qtr(S) \rightarrow C$. Path costs are monotone with respect to prefix order: if $\tau$ is a prefix of $\tau'$, then $pcost(\tau) \leq pcost (\tau')$. 
	
	\item \textbf{Probability distribution}: Given an arbitrary state $[s]$ in $LTS(S)$, we define a probability space $(\Omega, \mathcal{F}, \mathcal{P})$ on the set of possible transitions and their corresponding costs from this state, where the sample space $\Omega$ is the set of all transitions in $Q\times\Sigma\times Q$, the $\sigma$-algebra $\mathcal{F}$ is defined in the straightforward way based on the set of elementary events $\Omega$, and $\mathcal{P}$ is a probability measure $\mathcal{P} : \mathcal{F} \rightarrow [0, 1]$. 

	Importantly, this allows us to define, for any path, the notion of probability for costs incurred at each step. This probability space is readily extended for arbitrary finite paths, where we assume that the cost probabilities at each step are independent of previous steps, i.e., historyless. This process induces a Markov chain, whose state at each step is given by the state $[s]$ of the corresponding LTS, and whose transitions are LTS transitions, with costs and probabilities as above. 
\end{enumerate}

\paragraph{Distributional Linearizability.} With this in place, we now define distributionally linearizable data structures:
\begin{definition}
	\label{def:dl}
	Let $D$ be a randomized concurrent data structure, and let $R$ be a randomized quantitative relaxation $R$ of a sequential specification $S$ with respect to a cost function \lit{cost}, and a probability distribution $\mathcal{P}$ on costs. We say that $D$ is distributionally linearizable to $R$ iff for every concurrent schedule $\sigma$, there exists a mapping of completed operations in $D$ under $\sigma$ to transitions in the quantitative path of $R$, preserving outputs, and respecting the order of non-overlapping operations. This mapping can be used to associate any schedule $\sigma$ to a \emph{distribution of costs} for $D$ under the schedule $\sigma$.  
\end{definition}

We now make a few important remarks on this definition. 
\begin{enumerate}
\item The main difficulty when formally defining the ``costs" incurred by $D$ in a concurrent execution is in dealing with the execution history, and with the impact of pending operations on these costs. The above definition allows us to define costs, given a schedule, only in terms of the sequential process $R$, and bounds the incurred costs in terms of the probability distribution defined in $R$. 
This definition ensures that the probability distribution on costs incurred at each step only depends on the current state of the sequential process. 

\item The second key question is how to use this definition. One subtle aspect of this definition is that the mapping to the sequential randomized quantitative relaxation is done \emph{per schedule}: intuitively, this is because an adversary might change the schedule, and cause the distribution of costs of the data structure to \emph{change}. 
Thus, it is often difficult to specify a precise cost distribution, which covers all possible schedules. 
However, for the data structures we analyze, we will be able to provide \emph{tail bounds} on the cost distributions induced by \emph{all possible schedules}. 
\end{enumerate}

The natural next question, which we answer in the following section,  is whether non-trivial such data structures exist and can be analyzed. 

\section{Analysis of the MultiCounter}

\noindent We will focus on proving the following result. 

\begin{theorem}
\label{thm:main}
	Given an oblivious adversary, $n$ distributed counters and $n$ threads, for any fixed schedule, 
	the MultiCounter algorithm is distributionally linearizable to a randomized relaxed sequential counter process, which, at any step $t$, returns a value that is at most  $O( n \log^2 n)$ away from the number of increments applied up to $t$, both in expectation and with high probability in $n$. 
\end{theorem}

\noindent We emphasize that the relaxation guarantees are independent of the time $t$ at which the guarantee is examined, and that they would thus hold in infinite executions. 

\subsection{Modeling the Concurrent Process} 
\label{sec:prelim}
In the following, we will focus on analyzing executions formed exclusively of \lit{increment} operations, whose lower-level steps may be interleaved. 
(Adding \lit{read} operations at any point during the execution will be immediate.) 
We model the process as follows. First, we assume a schedule that is fixed by the adversary. 
For each thread $P_j$, and non-negative integers $j$, we consider a sequence of \lit{increment} operations $(\id{op}^{(j)}_i)$, each of which is defined by its starting time $s^{(j)}_i$, corresponding to the time when its first read step was scheduled, and completion time $f^{(j)}_i$, corresponding to the time when its update time is scheduled, such that $s^{(j)}_{i+1} > f^{(j)}_i$ for all $i,j$. (Recall that the scheduler defines a global order on individual steps.) At most $n$ operations may be active at a given time, corresponding to the fact that we only have $n$ parallel threads. 

For each operation $\id{op}_i$, we record its \emph{contention} $\ell_i$ as the number of distinct \lit{increment} operations scheduled between its start and end time. (Alternatively, we could define this quantity as the number of operations which complete in the time interval $(s_i, f_i)$.) Note that at most $n - 1$ distinct operations can be concurrent with $\id{op}_i$ at any given time, but the contention for a specific operation is potentially unbounded. 

We can rephrase the original process as follows. 
For each operation $\id{op}_i$, the adversary sets the time when it performs the first and its second read of counter values / bin weights, as well as 
its contention $\ell_i$, by scheduling other operations concurrently. The only constraint on the adversary is that not more than $n$ operations can be active at the same time. 

Since the adversary is oblivious, we notice that the update process is equivalent to the following: 
at the time when the update is scheduled, the thread executing the operation generates two uniform random indices $i$ and $j$, and is given values $x_i$ and $x_j$ for the two corresponding counters / bin weights, read at previous (possibly different) points in time. We will stick to the bin weight formulation from now on, with the understanding that the two are equivalent. 

The thread will then increment the weight of the bin with the smaller value read (among $x_i$ and $x_j$) by $1$. 
This formulation has the slight advantage that it makes the update process sequential, by moving the random choices to the time when the update is made, using the principle of deferred decisions. Critically, the bin weights on which the update decision is based are potentially stale. 
We will focus on this simplified variant of the process in the following.

\paragraph{Discussion.}
The key difference between the above process and the classic power-of-two-choices process is the fact that the choice of bin / counter which the thread updates is based on stale, potentially invalid information. Recall that key to the strong balancing properties of the classic process is the fact that it is biased towards inserting in \emph{less loaded} bins; the process which inserts into randomly chosen bins is known to diverge~\cite{PTW15}. 
In particular, notice it is possible that, by the time when the thread performs the update, the order of the bins' load may have changed, i.e. the thread in fact inserts into the \emph{more loaded} bin among its two choices at the time of the update. 

Since the oblivious adversary decides its schedule independently of the threads' random choices, it cannot \emph{deterministically} cause a specific update to insert into the more loaded bin. However, it can \emph{significantly bias} an update towards inserting into the more loaded bin: 

Assume for example an execution suffix where all $n$ threads read concurrently at some time $t_R$\footnote{Technically, since we count time in terms of shared-memory operations, these reads occur at consecutive times after $t_R$. However, all their read values are identical to the read value at $t_R$, and hence we choose to simplify notation in this way.} and then proceed to perform updates, one after another. 
Pick an operation $\id{op}$ for which the gap between the two values read $x_i$ and $x_j$ (at the time of the read) is $1$, say $x_i = x_j + 1$. 
So $\id{op}$ will increment $x_j$. At the same time, notice that all the other operations which read concurrently with $\id{op}$ are biased towards inserting in $x_j$ rather than $x_i$, since its rank (in increasing order of weight) is lower than that of bin $i$. 
Hence, as the adversary schedules more and more operations between $t_R$ at $\id{op}$'s update time, it is increasingly likely to \emph{invert} the relation between $i$ and $j$ by the time of $\id{op}$'s update, causing it to insert into the ``wrong" bin. 

The previous example suggests that the adversary is able to bias some subset of the operations towards picking the wrong bin at the time of the update. 
Another issue is that operations which experience high contention, for which there are many updates between the read point and the update point, the read values $x_i$ and $x_j$ become meaningless: for example, if the weights of bin $i$ and $j$ become equal at some time $t_0$ between $t_R$ and $\id{op}$'s update, then from this point in time these two bins appear completely symmetrical to the algorithm, and $\id{op}$'s choice given the information that $x_i > x_j$ at $t_R$ may be no better than uniform random. 

One issue which further complicates this last example is that, at $t_0$, there may be a non-zero number of other operations which already made their reads (for instance, at $t_R$), but have not updated yet. Since these operations read at a point where $x_i > x_j$, they are in fact biased towards inserting in $x_j$. 
So, looking at the event that $\id{op}$ updates the \emph{less loaded} of its two random choices at update time, we notice that its probability in this example is \emph{strictly worse} than uniform random choice. 

We summarize this somewhat lengthy discussion with two points, which will be useful in our analysis: 

\begin{enumerate}
	\item As they experience concurrent updates, operations may accrue bias towards inserting into the \emph{more loaded} of their two random choices. 
	\item Long-running operations may in fact have a higher probability of inserting into the more loaded bin than into the less loaded one, i.e. may become biased towards making the ``wrong" choice at the time of the update.  
\end{enumerate}

\subsection{Notation and Background} \label{sec:definitions}

\noindent For any bin $i$ and time $t$, let $x_i(t)$ be the weight of bin $i$ at time $t$
and let $x(t)=(x_1(t),x_2(t),...,x_n(t))$ be a vector of weights.
Let $\mu(t) = \sum_{i = 1}^{m} x_i(t) / n$ be the average weight at time $t$ over the bins. 
Let $\alpha < 1$ be a parameter to be fixed later. 
At each step $t+1$, instead of increment by 1 we allow increment $w(t)$ to be a random variable.
Even though we initially concentrate on the case with counters ($w(t)=1$), it is useful to 
prove several general Lemmas with random weights in mind, since
we will need to use them later.

Define  $$\Phi^{seq}(t) = \sum_{i = 1}^n e^{\alpha (x_i(t)-\mu(t))}, \,\text{and}\, \Psi^{seq}(t) = \sum_{i = 1}^n e^{ - \alpha (x_i(t)-\mu(t))}.$$  

\noindent Finally, define the potential function 
$$ \Gamma^{seq}(t) = \Phi^{seq}(t) + \Psi^{seq}(t).$$ 

We use superscript $seq$ to denote potential functions given by sequential process, which always increments
the counter with the smaller load (In our concurrent process this is not true).
In order to bound $\Gamma_{seq}$, $w(t)$ should have the following properties :

    \begin{equation} \label{weightproperty1}
    \E[w(t)]=1
    \end{equation}
    and
    there exist constants $S \ge 1$ and $\lambda>0$, such that for any $|x|\le \lambda/2$: 
    \begin{equation} \label{weightproperty2}
        \E[(e^{x w(t))})'']=\E[M''(x)]<2S.
    \end{equation}

In the case of counters ($w(t)=1$), we can use $\lambda=1$ and $S=1$ since $e^{\frac{1}{2}}\le 2$.

\noindent The main technical result of~\cite{PTW15} can be phrased as:
\begin{theorem}
\label{thm:beta}
    Let $\epsilon = \frac{1}{16}$ and let $\alpha \le \min{\left(\frac{\epsilon}{6S}, \frac{\lambda}{2}\right)}$ be a parameter as given above.
	Then there exists a constant $C(\epsilon) = \poly ( \frac{1}{\epsilon} )$ such that, for any time $t\geq 0$, we have $\E[ \Gamma^{seq}(t) ] \leq  \frac{4C(\epsilon)n}{\alpha\epsilon}$.
\end{theorem}

We would like to point out that the upper bound on $\alpha$ and the value of $\epsilon$ are chosen according
to the conditions required in~\cite{PTW15} and we will assume that they hold throughout this paper (Later on, we will assume even smaller upper bound on $\alpha$):
\begin{align} \label{alphaproperty}
    \alpha \le \frac{\epsilon}{6S} \text{ and }
   \alpha \le \frac{\lambda}{2}.
\end{align}
Our goal will be to prove similar theorem 
in the concurrent case.

Note that this implies that the maximum gap between the most loaded and the least loaded bin at a step is $2\frac{\log n}{\alpha}+
O\left(\frac{\log{\frac{1}{\alpha}}}{\alpha}\right)$ in expectation and with high probability in $n$ (As shown in \cite{PTW15}). 

The proof of the above theorem uses the following Lemma, which we also are going to rely on:

\begin{lemma} \label{lem:beta}
Let $\alpha$ and $\epsilon$ and $C(\epsilon)$ be the parameters defined in Theorem \ref{thm:beta}. Then for any step $t$:
\begin{align*}
\E[\Gamma^{seq}(t+1)|x(t)]\le \Big(1-\frac{\alpha\epsilon}{4n}\Big)\Gamma^{seq}(t)+C(\epsilon).
\end{align*}
\end{lemma}

\subsection{Naive Upper and Lower Bounds}
Let $\Gamma^{con}(t),\Phi^{con}(t)$ and $\Psi^{con}(t)$ be the potential functions in the concurrent case.
We start with proving the following result:

In this section we derive upper and lower bounds on $\Gamma^{con}$ per step.
These bounds just use the fact that for any bin $i$ the probability of incrementing it is at most $\frac{2}{n}$,
and this is true both for sequential and concurrent processes.

We assume that at step $t+1$, increment $w(t)$ satisfies conditions from Section 
\ref{sec:definitions}. We start with the upper bound:
\begin{lemma}
\label{lem:BadBoundUpper}
For any operation $op_t$
\begin{equation}
\E[\Gamma^{con}(t+1)|x(t)]\le \Bigg(1+
\frac{4 \alpha}{n}\Bigg)\Gamma^{con}(t).
\end{equation}
\end{lemma}

\begin{proof}
First we consider what is expected change in $\Phi^{con}$. Let $y_i=x_i(t)-\mu(t)$ and let $\Phi^{con}_i (t) = e^{\alpha y_i}$. Also, let $\Delta \Phi^{con}=\Phi^{con}(t+1)-\Phi^{con}(t)$ and 
$\Delta \Psi^{con}=\Psi^{con}(t+1)-\Psi^{con}(t)$
We have two cases to consider. If bin $i$ is chosen, then the change is:
\begin{align*}
\E&[\Delta \Phi^{con}_i|x(t)]  = \E[\Phi^{con}_i(t+1)|x(t)] - \Phi^{con}_i(t)  \\ 
& = \E\Bigg[\exp\Bigg(\alpha \Big( x_i(t)-\mu(t)+  (w(t) - \frac{w(t)}{n} ) \Big) \Bigg)\Bigg|x(t)\Bigg] - e^{\alpha y_i} \\ &= e^{\alpha y_i} 
\Bigg(\E\Big[\exp\Big(w(t) \alpha (1-\frac{1}{n})\Big)\Big]-1\Bigg) 
\\ &= e^{\alpha y_i} 
\Bigg(\E\Big[M\Big(\alpha(1-\frac{1}{n})\Big)\Big]-1\Bigg) 
\\ &\overset{(*)}{=}
e^{\alpha y_i} 
\Bigg(\E\Big[M(0)+M'(0)\alpha(1-\frac{1}{n})+M''(\xi)\alpha^2 (1-\frac{1}{n})^2\Big/2\Big]-1 \Bigg) \\&=
e^{\alpha y_i} 
\Bigg(\E\Big[1+w(t)\alpha(1-\frac{1}{n})+M''(\xi)\alpha^2 (1-\frac{1}{n})^2\Big/2\Big]-1 \Bigg)
\\&\overset{(\ref{weightproperty1}),(\ref{weightproperty2})}{\le} e^{\alpha y_i} 
\Bigg(\alpha(1-\frac{1}{n})+S\alpha^2(1-\frac{1}{n})^2\Bigg)
\end{align*}
Where in~$(*)$ we used the tailor expansion of $M(x)$  around $0$ and in the last step we used that 
$0 \le \xi \le \alpha(1-\frac{1}{n}) \le \frac{\lambda}{2}.$

Using  similar arguments we can prove that, when some other bin $i$ is not chosen:

\begin{equation*}
\Delta \Phi^{con}_i \le e^{\alpha y_i}\Bigg(-\frac{\alpha}{n}+S\frac{\alpha^2}{n^2}\Bigg) \le 0.
\end{equation*}

Let $p_i \le 2/n$ be the probability of bin $i$ being chosen for increment.
We get that:
\begin{align*}
	\E\Big [ \Delta \Phi^{con}_i | x(t)\Big] &\le p_i e^{\alpha y_i} \Bigg(\alpha(1-\frac{1}{n})+S\alpha^2(1-\frac{1}{n})^2\Bigg) \\&\le
	p_i (\alpha+S\alpha^2) \le \frac{4\alpha}{n} e^{\alpha y_i}.
\end{align*}
Hence:
\begin{equation}\label{PhiBound}
\E[\Delta \Phi^{con} |x(t)]=\sum_{i=1}^{n} \E[\Delta \Phi_i^{con} |x(t)] \le 
\frac{4\alpha}{n}\Phi^{con}(t).
\end{equation}
In a similar way, we can prove that:
\begin{align*}
\E[\Delta \Psi^{con} |x(t)] &\le \sum_{i=1}^{n} (1-p_i) (\frac{\alpha}{n}+\frac{S\alpha^2}{n^2})e^{-\alpha y_i}
\\& \le \sum_{i=1}^{n} (\frac{\alpha}{n}+\frac{S\alpha^2}{n^2})e^{-\alpha y_i} \le \frac{4 \alpha}{n} \Psi^{con}(t)
\end{align*}
Combining this with inequality (\ref{PhiBound}), and using the definitions of $\Delta \Phi^{con}$ and 
$\Delta \Psi^{con}$
gives us proof of the Lemma. 
\end{proof}

We proceed by showing the lower bound:

\begin{lemma}
\label{lem:BadBoundLower}
For any operation $op_t$
\begin{equation}
\E[\Gamma^{con}(t+1)|x(t)]\ge \Bigg(1-
\frac{2 \alpha}{n}\Bigg)\Gamma^{con}(t).
\end{equation}
\end{lemma}

\begin{proof}
First we consider what is expected change in $\Phi^{con}$. Let $y_i=x_i(t)-\mu(t)$ and let $\Phi^{con}_i (t) = e^{\alpha y_i}$. We have two cases here. If bin $i$ is chosen, then as in the previous lemma the change is:
\begin{align*}
\E&[\Delta \Phi^{con}_i|x(t)]  = \E[\Phi^{con}_i(t+1)|x(t)] - \Phi^{con}_i(t) \\&=
e^{\alpha y_i} 
\Bigg(\E\Big[1+w(t)\alpha(1-\frac{1}{n})+M''(\xi)\alpha^2 (1-\frac{1}{n})^2\Big/2\Big]-1 \Bigg)
\\&\ge e^{\alpha y_i} 
\alpha(1-\frac{1}{n}) \ge \frac{\alpha}{2} e^{\alpha y_i} \ge 0.
\end{align*}
Where in the last step we used that $n \ge 2$ and the fact that exponential function is non-negative.

Using  similar arguments we can prove that, when some other bin $i$ is not chosen:

\begin{equation*}
\Delta \Phi^{con}_i \ge -\frac{\alpha}{n} e^{\alpha y_i}.
\end{equation*}

Let $p_i \le 2/n$ be the probability of bin $i$ being chosen for increment
We get that:
\begin{align*}
	\E\Big [ \Delta \Phi^{con}_i | x(t)\Big] &\ge -(1-p_i) \frac{\alpha}{n} e^{\alpha y_i} \ge -\frac{\alpha}{n} e^{\alpha y_i}.
\end{align*}
Hence:
\begin{equation}\label{PhiBoundUpper}
\E[\Delta \Phi^{con} |x(t)]=\sum_{i=1}^{n} \E[\Delta \Phi_i^{con} |x(t)] \ge 
-\frac{\alpha}{n}\Phi^{con}(t).
\end{equation}
In a similar way, we can prove that:
\begin{align*}
\E[\Delta \Psi^{con} |x(t)] &\ge -\sum_{i=1}^n p_i \alpha (1-\frac{1}{n}) \ge -\frac{2 \alpha}{n} \Psi^{con}(t)
\end{align*}
Combining this with inequality (\ref{PhiBoundUpper}), and using the definitions of $\Delta \Phi^{con}_i$ and 
$\Delta \Psi^{con}_i$
gives us proof of the Lemma. 
\end{proof}

\subsection{Main Argument}

\noindent Now we consider $Cn$ ($C$ is a constant which we will fix later) consecutive operations and prove that at most $n$ of them can be bad:

\begin{lemma} \label{lem:ConsOps}
For any $t$, we have that $|t':t \le t' \le t+Cn-1,  \ell_{t'} > Cn| < n$. 
\end{lemma}
\begin{proof}
We argue by contradiction. Let us assume that the number of bad operations is at least $n$. By the pigeonhole principle, there exist bad  operations $op_i$ and $op_j$, $t \le i < j \le t+Cn-1$, which are performed by the same thread. This means that since these operations are not concurrent, we have that $s_j > f_i=i$. Thus, we get a contradiction: $Cn \le \ell_j = |t': s_j \le t' < f_j=j| \le j-i < Cn$.
\end{proof}

We call operation $op_t$ \textbf{good} if $\ell_t \le Cn$, otherwise we call it \textbf{bad}.
For each bin $i$, and step $t \ge Cn$, let $H_i(t)$ be the number of times $i$ was chosen by operations $op_{t-Cn+1}$, $op_{t-Cn+1}$, ...
$op_{t}$. In this case, if some operations chooses $i$ and $j$, we count both as chosen and we also say that $i$ was chosen twice if $i=j$. 
Observe that if $op_{t+1}$ is good, then $H_i(t)$ is the upper bound on the number of increments bin $i$
receives during the entire run of operation $op_{t+1}$ (excluding the increment which might be performed by $op_{t+1}$).
Also, let $H_{max}(t)=max\{H_1(t), H_2(t), ..., H_n(t)\}$.
Now we are ready to bound the potential.

\paragraph{Upper Bound on a Potential for Counters}

We concentrate on the case when $w(t)=1$, for any $t$ (The case with counters).

\begin{lemma} \label{lem:gammmadifference}
For any \textbf{good} operation $op_{t+1}$, such that $t \ge Cn$:
\begin{align*}
\E[\Gamma^{con}(t+1)|x(t),H_{max}(t)]\le \Big(1-\frac{\alpha\epsilon}{4n}\Big)\Gamma^{con}(t)&+C(\epsilon)
\\&+\frac{4\alpha \Gamma^{con}(t)}{n} (e^{\alpha H_{max}(t)}-1).
\end{align*}
\end{lemma}

\begin{proof}
Since we condition on $x(t)$, we can assume that potentials at step $t$ are the same both for sequential and concurrent processes (This is not true 
for the next step since  since processes can increment different bins).
We have that:
\begin{align*}
\E&[\Gamma^{con}(t+1)|x(t),H_{max}(t)]\\&=\E[\Gamma^{seq}(t+1)|x(t),H_{max}(t)]+\E[\Gamma^{con}(t+1)-\Gamma^{seq}(t+1)|x(t),H_{max}(t)]
\\&\overset{\text{Lemma }\ref{lem:beta}}{\le} \Big(1-\frac{\alpha\epsilon}{4n}\Big)\Gamma^{seq}(t)+C(\epsilon)
\\&\quad\quad\quad\quad\quad\quad\quad\quad\quad\quad\quad\quad+
\E[\Gamma^{seq}(t+1)-\Gamma^{con}(t+1)|x(t),H_{max}(t)].
\end{align*}
Hence our goal is to upper bound  $\E[\Gamma^{seq}(t+1)-\Gamma^{con}(t+1)|x(t),H_{max}(t)]$.
w.l.o.g we assume that $x_1(t) \le x_2(t) ... \le x_n(t)$.
We couple sequential and concurrent processes so that the bin choices $i$ and $j$ are the same in both cases.
Let $i \le j$, then sequential process always increments bin $i$, but for the concurrent process it depends 
on when its reads occurred, for example it can be that during reads the load of bin $j$ was smaller than
load of bin $i$ but then the increments done by concurrent processes reversed the order.
The crucial thing is that in this case $x_j(t)-x_i(t) \le H_{max}(t)$.
Hence, assuming the worst case (concurrent process increments bin $j$) we have that 
\begin{align*}
\Phi^{con}(t+1)-\Phi^{seq}(t+1)&=e^{\alpha(x_j(t)-\mu(t)+1-\frac{1}{n})}+e^{\alpha(x_i(t)-\mu(t)-\frac{1}{n})} 
\\&\quad\quad\quad-
e^{\alpha(x_i(t)-\mu(t)+1-\frac{1}{n})}-e^{\alpha(x_j(t)-\mu(t)-\frac{1}{n})}  \\&=
e^{\alpha(x_i(t)-\mu(t))}e^{-\frac{\alpha}{n}}(e^{\alpha}-1)(e^{\alpha(x_j(t)-x_i(t))}-1) \\&\le
2\alpha e^{\alpha(x_i(t)-\mu(t))} (e^{\alpha H_{max}(t)}-1).
\end{align*}
Where in the last step we used that $e^{\alpha} \le 1+2\alpha$, since $\alpha \le \frac{1}{2}$.
Also, 
\begin{align*}
\Psi^{con}(t+1)-\Psi^{seq}(t+1)&=e^{-\alpha(x_j(t)-\mu(t)+1-\frac{1}{n})}+e^{-\alpha(x_i(t)-\mu(t)-\frac{1}{n})} 
\\&\quad\quad\quad-
e^{-\alpha(x_i(t)-\mu(t)+1-\frac{1}{n})}-e^{-\alpha(x_j(t)-\mu(t)-\frac{1}{n})} \\&=
e^{-\alpha(x_i(t)-\mu(t))}e^{\frac{\alpha}{n}}(e^{-\alpha}-1)(e^{-\alpha(x_j(t)-x_i(t))}-1) \\&=
e^{-\alpha(x_i(t)-\mu(t))}e^{\frac{\alpha}{n}}(1-e^{-\alpha})(1-e^{-\alpha(x_j(t)-x_i(t))}) \\&\le
2\alpha e^{-\alpha(x_i(t)-\mu(t))}(1-e^{-\alpha H_{max}(t)}).
\end{align*}
Where in the last step we used that $e^{\frac{\alpha}{n}} \le 2$, since $\alpha \le \frac{1}{2}$.
The above bounds no longer depend on $j$ and for any bin $i$ the probability of being one out of two random choices 
of $op_t$ is at most $\frac{2}{n}$, hence:
\begin{align*}
\E&[\Gamma^{seq}(t+1)-\Gamma^{con}(t+1)|x(t)] \\&\le \sum_{i=1}^n \frac{4\alpha}{n} e^{\alpha(x_i(t)-\mu(t))} (e^{\alpha H_{max}(t)}-1) 
+
\sum_{i=1}^n \frac{4\alpha}{n} e^{-\alpha(x_i(t)-\mu(t))} (1-e^{-\alpha H_{max}(t)}) \\&=
\frac{4\alpha \Phi^{con}(t)}{n} (e^{\alpha H_{max}(t)}-1) 
+
\frac{4\alpha \Psi^{con}(t)}{n} (1-e^{-\alpha H_{max}(t)})
\\&\le \frac{4\alpha \Gamma^{con}(t)}{n} (e^{\alpha H_{max}(t)}-1).
\end{align*}
Where in the last step we used that $e^{\alpha H_{max}(t)}+e^{-\alpha H_{max}(t)} \ge 2$.
\end{proof}

Let $N=\lfloor\frac{2Cn}{2e^3C \log{n}} \rfloor$ and for $0 \le K \le N$, let $A_K(t)$ be the event that 
$2e^3C K \log{n} \le H_{max}(t) < 2e^3C (K+1) \log{n}$. We proceed by proving the following lemma:

\begin{lemma} \label{lem:refinedgammadifference}
For any \textbf{good} operation $op_{t+1}$, such that $t \ge Cn$:
\begin{align*}
\E[\Gamma^{con}(t+1)]&\le \Big(1-\frac{\alpha\epsilon}{4n}\Big)\E[\Gamma^{con}(t)]+C(\epsilon)
\\&+\sum_{K=0}^N
\frac{4\alpha \E[\Gamma^{con}(t)|A_K(t)]Pr[A_K(t)]}{n} (e^{2\alpha e^3C (K+1)
\log{n}}-1).
\end{align*}
\end{lemma}

\begin{proof}
First we remove conditioning on $x(t)$:
\begin{align*}
\E&[\Gamma^{con}(t+1)|H_{max}(t)]=\E_{x(t)|H_{max}(t)}[\E[\Gamma^{con}(t+1)|x(t),H_{max}(t)]] \\ &\le
\Big(1-\frac{\alpha\epsilon}{4n}\Big)\E[\Gamma^{con}(t)|H_{max}(t)]+C(\epsilon)
\\&\quad\quad\quad\quad\quad\quad\quad\quad\quad\quad\quad\quad\quad\quad+\frac{4\alpha \E[\Gamma^{con}(t)|H_{max}(t)]}{n} (e^{\alpha H_{max}(t)}-1).
\end{align*}
Next, we remove conditioning on $H_{max}(t)$ from the left side of the above inequality. Using Lemma \ref{lem:gammmadifference} we get that:
\begin{align*}
\E&[\Gamma^{con}(t+1)]=\sum_{K=0}^{N-1} \E[\Gamma^{con}(t+1)|A_K(t)]Pr[A_K(t)] \\ &\le \sum_{K=0}^N \Bigg(
\Big(1-\frac{\alpha\epsilon}{4n}\Big)\E[\Gamma^{con}(t)|A_K(t)]Pr[A_K(t)]+C(\epsilon)Pr[A_K(t)]
\\&\quad\quad\quad\quad\quad\quad\quad\quad\quad+\frac{4\alpha \E[\Gamma^{con}(t)|A_K(t)]Pr[A_K(t)]}{n} (e^{2\alpha e^3C (K+1)
\log{n}}-1)\Bigg)
\\ &= \Big(1-\frac{\alpha\epsilon}{4n}\Big)\E[\Gamma^{con}(t)]+C(\epsilon) \\&\quad\quad\quad\quad\quad\quad\quad+\sum_{K=0}^N
\frac{4\alpha \E[\Gamma^{con}(t)|A_K(t)]Pr[A_K(t)]}{n} (e^{2\alpha e^3C (K+1)
\log{n}}-1).
\end{align*}
\end{proof}

Our next goal is to upper bound $Pr[A_K(t)]$ and $\E[\Gamma^{con}(t)|A_K(t)]$
For this start with deriving the concentration bounds for $H_{max}(t)$.

\begin{lemma} \label{lem:Chernoff}
For any $t > Cn$ and constant $K \ge 1$:
\begin{align*}
Pr[A_K(t)]\le Pr[H_{max}(t) \ge 2K e^3C \log{n}] \le  \frac{1}{(eK \log{n})^{2 K C e^3 \log{n}}}.
\end{align*}
\end{lemma}

\begin{proof}
Note that $H_{max}(t)$ is a maximum number of balls some bin receives if we throw $2Cn$ balls into $n$ initially empty bins (Recall that all the random choices which operations make are independent). 
For a fixed bin $i$, let $H_i(t)$ be the number of balls it receives.
We know that $\E[H_i(t)]=2C$.
Hence, using Chernoff's inequality we get that 
\begin{align*}
    Pr [H_i(t) \ge 2K e^3 C \log{n}] &\le \Bigg(\frac{e^{K e^3 \log{n} - 1}}{(K e^3 \log{n})^{K e^3 \log{n}}}\Bigg)^{2C}
    \\&\le \frac{1}{n} \frac{1}{(e K \log{n})^{2K C e^3 \log{n}}}.
\end{align*}
By union bounding over $n$ bins we get the proof of the lemma.
\end{proof}

We proceed by upper bounding $\E[\Gamma^{con}(t)|A_K(t)]$.

\begin{lemma} \label{lem:lowerboundcond}
For any $t \ge Cn$:
\begin{align*}
\E[\Gamma^{con}(t)|A_K(t)] \le \E[\Gamma^{con}(t)]e^{3\alpha e^3C (K+1) \log{n}}.
\end{align*}
\end{lemma}
\begin{proof}
\begin{align*}
\E&[\Phi^{con}(t)|A_K(t),x(t-Cn)]-\Phi^{con}(t-Cn)\\&=
\sum_{i=1}^n e^{\alpha (x_i(t-Cn)-\mu(t-Cn))}\Bigg(e^{\alpha\Big(x_i(t)-\mu(t)-x_i(t-Cn)+\mu(t-Cn)\Big)}-1\Bigg).
\end{align*}
Since we condition on $A_k(t)$, for every $i$ we have that 
\begin{align*}
x_i(t)-x_i(t-Cn) \le H_{max}(t) \le 2e^3C (K+1) \log{n}.
\end{align*}
Also, $\mu(t) > \mu(t-Cn)$.
Thus
\begin{align*}
\E&[\Phi^{con}(t)|A_K(t),x(t-Cn)]-\Phi^{con}(t-Cn) \\&\le
\sum_{i=1}^n e^{\alpha (x_i(t-Cn)-\mu(t-Cn))}\Bigg(e^{2\alpha e^3C (K+1) \log{n}}-1\Bigg)
\\&=\Phi^{con}(t-Cn)\Bigg(e^{2\alpha e^3C (K+1) \log{n}}-1\Bigg).
\end{align*}
Similarly
\begin{align*}
\E&[\Psi^{con}(t)|A_K(t),x(t-Cn)]-\Psi^{con}(t-Cn)\\&=
\sum_{i=1}^n e^{-\alpha (x_i(t-Cn)-\mu(t-Cn))}\Bigg(e^{-\alpha\Big(x_i(t)-\mu(t)-x_i(t-Cn)+\mu(t-Cn)\Big)}-1\Bigg).
\end{align*}
We have that $x_i(t) \ge x_i(t-Cn)$ and 
\begin{align*}
\mu(t) - \mu(t-Cn) \le H_{max}(t) \le 2e^3C (K+1) \log{n}
\end{align*}
Thus
\begin{align*}
\E&[\Phi^{con}(t)|A_K(t),x(t-Cn)]-\Phi^{con}(t-Cn) \\&\le
\sum_{i=1}^n e^{-\alpha (x_i(t-Cn)-\mu(t-Cn))}\Bigg(e^{2\alpha e^3C (K+1) \log{n}}-1\Bigg)
\\&=\Psi^{con}(t-Cn) \Bigg(e^{2\alpha e^3C (K+1) \log{n}}-1\Bigg).
\end{align*}
Hence
\begin{align*}
\E&[\Gamma^{con}(t)|A_K(t),x(t-Cn)]-\Gamma^{con}(t-Cn)] \\&\le \Gamma^{con}(t-Cn)\Bigg(e^{2\alpha e^3C (K+1) \log{n}}-1\Bigg)
\end{align*}
Notice that $A_K(t)$ is independent of $x(t-Cn)$, since in the definition of $H_{max}(t)$ we just consider random choices made by $op_{t-Cn+1},...,op_t$.
This allows us to remove conditioning on $x(t-Cn)$ and after regrouping the terms in the above inequality we get
\begin{align} \label{eqn:condgammalower}
\E&[\Gamma^{con}(t)|A_K(t)] \le \E[\Gamma^{con}(t-Cn)]e^{2\alpha e^3C (K+1) \log{n}}
\end{align}
By applying Lemma \ref{lem:BadBoundLower} $Cn$ times we get that 
\begin{align*}
\E[\Gamma^{con}(t)|x(t-Cn)] \ge \Gamma^{con}(t-Cn)\Bigg(1-
\frac{2 \alpha}{n}\Bigg)^{Cn} \ge \Gamma^{con}(t-Cn)e^{-4C\alpha}.
\end{align*}
After removing conditioning we get that 
\begin{align*}
\E[\Gamma^{con}(t)] \ge \E[\Gamma^{con}(t-Cn)]e^{-4C\alpha}.
\end{align*}
By combining the above inequality with (\ref{eqn:condgammalower}) we get that:
\begin{align*}
\E[\Gamma^{con}(t)|A_K(t)] &\le \E[\Gamma^{con}(t)]e^{2\alpha e^3C (K+1) \log{n}}e^{4\alpha C}
\\&\le \E[\Gamma^{con}(t)]e^{3\alpha e^3C (K+1) \log{n}}.
\end{align*}
\end{proof}

Finally
\begin{lemma} \label{lem:goodbound}
For any good operation $op_{t+1}$, such that $t \ge Cn$, we have that if $C \ge 2$ and $\alpha\le \frac{1}{4096 C e^3\log n}$ then
\begin{align*}
\E[\Gamma^{con}(t+1)]&\le \Big(1-\frac{\alpha\epsilon}{8n}\Big)\E[\Gamma^{con}(t)]+C(\epsilon).
\end{align*}
\end{lemma}

\begin{proof}
Since $\alpha \le \frac{1}{4096 C e^3\log n}$ and $C \ge 2$:
\begin{align} \label{eqn:needtomodify}
\sum_{K=1}^N
&\E[\Gamma^{con}(t)|A_K(t)]Pr[A_K(t)] (e^{2\alpha e^3C (K+1)\log{n}}-1)
\nonumber \\&\overset{\text{Lemmas \ref{lem:Chernoff} and \ref{lem:lowerboundcond}}}{\le} \sum_{K=1}^N  \frac{\E[\Gamma^{con}(t)]e^{5\alpha e^3C (K+1) \log{n}}}
{(eK \log{n})^{2 K C e^3 \log{n}}} \nonumber \\ &\le \sum_{K=1}^N  \frac{\E[\Gamma^{con}(t)]e^{e^3C K \log{n}}}
{e^{2 K C e^3 \log{n}}} \le \sum_{K=1}^N  \frac{\E[\Gamma^{con}(t)]}
{e^{2K e^3 \log{n}}} \le \sum_{K=1}^{\infty}  \frac{\E[\Gamma^{con}(t)]}
{n^{16K}} \nonumber \\&\le \frac{2\E[\Gamma^{con}(t)]}{n^{16}} \le \frac{\E[\Gamma^{con}(t)]}{2048}.
\end{align}
Also, for $K=0$
\begin{align*}
\E&[\Gamma^{con}(t)|A_0(t)]Pr[A_0(t)] (e^{2\alpha e^3C \log{n}}-1)
\\&\overset{\text{Lemma \ref{lem:lowerboundcond}}}{\le} 
\E[\Gamma^{con}(t)]e^{3\alpha e^3C \log{n}}(e^{2\alpha e^3C \log{n}}-1) \\
&\le\E[\Gamma^{con}(t)]e^{\frac{3}{4096}}(e^{\frac{1}{2048}}-1) \\&\le 2\E[\Gamma^{con}(t)] \frac{1}{1024} = \frac{\E[\Gamma^{con}(t)]}{512}.
\end{align*}
Hence, we get that 
\begin{align*}
\sum_{K=0}^N
\E&[\Gamma^{con}(t)|A_K(t)]Pr[A_K(t)] (e^{2\alpha e^3C (K+1)\log{n}}-1)
\\ &\le \frac{\E[\Gamma^{con}(t)]}{2048}+\frac{\E[\Gamma^{con}(t)]}{512} =\frac{5\E[\Gamma^{con}(t)]}{2048}.
\end{align*}
By plugging the above inequality in Lemma \ref{lem:refinedgammadifference} we get that 
\begin{align*}
\E[\Gamma^{con}(t+1)]&\le \Big(1-\frac{\alpha\epsilon}{4n}\Big)\E[\Gamma^{con}(t)]+C(\epsilon)+
\frac{20\alpha \E[\Gamma^{con}(t)]}{2048}.
\end{align*}
Recall that $\epsilon=\frac{1}{12}$, thus $\frac{20}{2048} \le \frac{1}{8\epsilon}$ and this finishes the proof of the lemma.
\end{proof}

\paragraph{Endgame.}
With all this machinery in place, we proceed to prove the following. 
\begin{lemma}
\label{lem:gamma-bound}
	If $\alpha \le \frac{1}{4096 C e^3\log n}$ and $C \ge 433$, then at any time step $t$
	$$\E[\Gamma^{con}(t)] \leq \frac{146 C(\epsilon)n}{\alpha \epsilon}.$$
\end{lemma}
\begin{proof}
	We will proceed by induction on $t$. We will first prove that, if $\E[\Gamma^{con}(t)] \leq \frac{146 C(\epsilon)n}{\alpha \epsilon}$ for  $t \ge Cn$, then $\E[\Gamma^{con}(t + C n) | \Gamma^{con}(t)] \leq  \frac{146 C(\epsilon)n}{\alpha \epsilon}$. 
	
	We have two cases. 
	The first is if there exists a time $\tau \in [t, t + Cn]$ such that $\E[\Gamma^{con}(\tau)] \leq \frac{72 C(\epsilon)n}{\alpha \epsilon}$. 
	Let us now focus on bounding the maximum expected value of $\Gamma^{con}(t + Cn)$ in this case. 
	First notice that the maximum expected increase of $\Gamma^{con}$ because of a good step is an additive $C(\epsilon)$ factor. By Lemma $\ref{lem:BadBoundUpper}$ The  expected value of $\Gamma^{con}$ after a bad operation is upper bounded a multiplicative $(1+\frac{4\alpha}{n})$ factor. Hence, by Lemma~\ref{lem:ConsOps} and  expected maximum value of $\Gamma^{con}$ at $t + Cn$ is at most 
\begin{align*}
	\left( \frac{72C(\epsilon)n}{\alpha \epsilon} + C(\epsilon) (C - 1) n \right) \left(1+
\frac{4\alpha}{n}\right)^n &\leq \left( \frac{72C(\epsilon)n}{\alpha \epsilon} + \frac{C(\epsilon) n}{4096 \alpha e^3\log n}\right)e^{4\alpha}
\\& \le \frac{146 C(\epsilon)}{\alpha \epsilon}.
\end{align*}

The second case is if there exists no such time in $[t, t + Cn]$, meaning that $\E[\Gamma^{con}(\tau)] > \frac{72 C(\epsilon)n}{\alpha \epsilon}, \forall \tau \in [t, t + Cn].$ Then, by Lemma~\ref{lem:goodbound}, we have that, at each good step, 
\begin{equation}
    \E[\Gamma^{con}(t+1)]\le \E[\Gamma^{con}(t)]\Big(1-\frac{\alpha \epsilon}{9n}\Big).
\end{equation}

\noindent Hence, we can expand the recursion to upper bound the change in $\Gamma^{con}$ between $t$ and $t + Cn$ as
\begin{align*}
    \E[\Gamma^{con}(t+Cn)] &\le \E[\Gamma^{con}(t)] \Big(1-\frac{\alpha \epsilon}{9n}\Big)^{(C-1)n}\left(1+
\frac{4\alpha}{n}\right)^n \\&\le \E[\Gamma^{con}(t)] e^{-\frac{\alpha \epsilon (C-1)}{9}+4\alpha} \le \E[\Gamma^{con}(t)].
\end{align*}
Where in the last step we used that $C \ge 1+36/\epsilon=433$.

To establish the base of induction note that by Lemma $\ref{lem:BadBoundUpper}$, for each $0 \le t \le 2Cn$:
\begin{align*}
\Gamma^{con}(t) &\le \Gamma^{con}(0) (1+\frac{4\alpha}{n})^{2Cn}=2n (1+\frac{4\alpha}{n})^{2Cn}
\\&\le 2ne^{8\alpha C} \le 4n \le \frac{146 C(\epsilon)n}{\alpha \epsilon}.
\end{align*}

\noindent This concludes the proof of the Lemma. 
\end{proof}

\noindent The following claim completes the proof of Theorem~\ref{thm:main}.  
\begin{lemma}
\label{lem:main}
	 Given an oblivious adversary, $n$ distributed counters and $n$ threads, for any time $t$ in the execution of the approximate counter algorithm the counter returns a value that is at most $O( n \log^2 n)$ away from the number of increment operations which completed up to time $t$, in expectation.
	Moreover, for any $t$ and all $R$ sufficiently large, we have
	\[
	\Pr \left[ \exists i: \left| n \cdot x_i (t) - n \cdot \mu_i (t) \right| > R n \log^2 n \right] \leq n^{-\Omega(R)} \; .
	\]
\end{lemma}
\begin{proof}
	The proof is similar to~\cite{PTW15} (the main difficulty was to reach asymptotically the same potential upper bound). 
	We aim to bound $\id{Gap}(t)$, the maximum gap between the weight of two bins at a step. 

By choosing $C=433$  and $\alpha=\frac{1}{4096 C e^3\log n}=\Theta(\frac{1}{\log n})$ and applying Lemma \ref{lem:gamma-bound} we get that $\E [\Phi^{con} (t)] = O(n \log n)$ and $\E [\Psi^{con} (t)] = O(n \log n)$ for all $t$.
Let $x_{max}(t)$ denote the maximum weight of any bin at time $t$, and let $x_{min}(t)$ be the minimum weight of any bin.
Then, we have
\begin{align*}
\alpha \E [x_{max}(t) - \mu (t)] &= \log \exp \left( \E [ \alpha (x_{max}(t) - \mu(t)) ] \right) \\
&\stackrel{(a)}{\leq} \log \E [\exp (\alpha (x_{max}(t) - \mu(t)))] \\
&\stackrel{(b)}{\leq} \log \E [\Phi^{con} (t)] \leq O(\log n + \log\log n)= O(\log n)\; ,
\end{align*}
where (a) follows from Jensen's inequality, and (b) follows from the definition of $\Phi^{con}$.
Similarly, we have $\E[\mu(t) - x_{min}(t)] \leq O(\log^2 n)$.
Since the true value of the counter at time $t$ is $n \cdot \mu(t)$, these two statements imply that for all $i$, we have $\E [|n \cdot x_i (t) - n \cdot \mu(t)|] \leq O(n \log^2 n)$, as desired.

We now prove the high probability bound.
Observe that if $\max (t) - \mu (t) > R \log^2 n$, then we have $\Gamma^{con} (t) \geq \Phi^{con} (t) \geq e^{\alpha R \log^2 n}$.
Hence, for large enough $R$:
\begin{align*}
\Pr [\max (t) - \mu (t) > R \log^2 n] &\leq \Pr [\Phi^{con} (t) \geq e^{\alpha R \log^2 n}] \\
&\overset{Markov}{\leq} \frac{O(n \log n)}{e^{\alpha R \log^2 n}} \\
&\leq n^{-O(R)} \; .
\end{align*}
Similarly,
$\Pr [\mu (t) - \min (t) > n \log n] \leq n^{-\Omega (R)}.$

Combining these two guarantees with a union bound immediately yields the desired guarantee.
\end{proof}

\section{Distributional Linearizability for Concurrent Relaxed Queues} 

We now extend the analysis in the previous section to imply distributional linearizability guarantees in concurrent executions for a variant of the MultiQueue process analyzed by~\cite{AKLN17}. 
This process is presented in Algorithm 2. We note that this process applies specifically to implement general concurrent \emph{queues}, and will also apply to \emph{priority queues} assuming that a sufficiently large buffer of elements always exists in the queues such that no insertion is ever performed on an element of \emph{higher} priority than an element which has already been removed. 

\subsection{Application to Concurrent Relaxed Queues}
\label{sec:queues}

\paragraph{Description.} 
We wish to implement a concurrent data structure with queue like semantics, so that we have guarantees on the rank of dequeued elements.
We assume we are given a set of $n$ linearizable priority queues such that each supports $\lit{Add} (e, p)$, $\lit{DeleteMin}$, $\lit{ReadMin}$, where $p$ is the priority of the element, and $\lit{ReadMin}$ returns the element with smallest priority in the priority queue, but does not remove it.
We also assume that each processor $i$ has access to a clock $\id{Clock}_i$ which gives an absolute time, and which are consistent amongst all the processors, that is, if processor $i$ reads $\id{Clock}_i$ in the linearization before processor $j$ reads $\id{Clock}_j$, then processor $i$'s value is smaller.
Such an assumption is realistic; recent Intel processors support the RDTSC hardware operation, which provides this functionality for cores on the same socket.

The procedure, given formally in Algorithm \ref{algo:approx-queue}, is similar to our approximate counter. 
To enqueue, a thread reads the wall clock, chooses a random priority queue, and adds the element to that priority queue with priority given by the time. 
To dequeue, we choose two random priority queues, find the one having a higher priority element on top, and delete from that priority queue. In case two processes enqueue to the same priority queue concurrently, their clock values will ensure a consistent ordering, handled by the internal implementation of the priority queues. 

\begin{algorithm}[ht]
\caption{Pseudocode for Relaxed Queue Algorithm.}
 \label{algo:approx-queue}
\begin{algorithmic}
\STATE \textbf{Shared}: $\id{PQs}[n]$ \textit{// Set of $n$ distinct priority queues}
\STATE \textbf{individual}: $\id{Clock}_i$ \textit{// A wall clock for processor $i$, for each $i$}

	\STATE \textbf{function} \lit{Enqueue}($e$)
	\STATE $p \gets \lit{Clock}_i.\lit{Read}()$
	\STATE $i \gets \lit{random}(1, n)$
	\STATE  $\id{PQs}[i].\lit{Add}(e, p)$
	\STATE
	\STATE \textbf{function} \lit{Dequeue}( )
	\STATE $i \gets \lit{random}(1, n)$
	\STATE $j \gets \lit{random}(1, n)$
	\STATE $(e_i, p_i) \gets \id{PQs}[i].\lit{ReadMin}()$
	\STATE $(e_j, p_j) \gets \id{PQs}[j].\lit{ReadMin}()$
	\STATE \lit{if} ( $ p_i > p_j$ ): $i = j$
	\STATE \textbf{return} $\id{PQs}[i].\lit{DeleteMin}()$
	
\end{algorithmic}
\end{algorithm}

\paragraph{Analysis.} 
The Analysis mostly follows the steps in \cite{AKLN17}.
We define the rank of element with timestamp $p$ as the number of elements which are currently in the system and 
have timestamp with value at most $p$ (Including itself, and assuming that no two operations have the same timestamp).

First we assume that Dequeues operations never see an empty queue. Given this assumption we can also assume:
\begin{itemize}
    \item Enqueue operations happen sequentially, sorted by linearization order
    \item Dequeue operations are invoked after all Enqueue operations are finished
\end{itemize}
Since the timestamps are increasing in linearization order, the two assumptions above do not change the outcome (the rank of returned element) of Dequeue operations and are needed solely for the purpose of analysis. 

We proceed by defining the auxiliary exponential label process. 
We are given $n$, initially empty queues in which we insert infinitely many labels as follows: for each queue $i$, if the 
last inserted label in $i$ is $v_i$ ($0$ if the queue is empty), then we insert label $v_i+Exp(\frac{1}{n})$ in it.
We define the rank of label $v$ as the number of labels which are currently in the queues and have value at most $v$ (Since exponential distribution is continuous we assume that no labels have the same value).
We will call these queues label queues to distinguish them from queues we use in Algorithm \ref{algo:approx-queue},

Theorem 2 in \cite{AKLN17} says that for any rank $r$ and queue $i$, probability of label with rank $r$ being 
in queue $i$ is $\frac{1}{n}$. The proof uses the memorylessness of exponential distribution.
Since for each queue the probability of element $e$ with initial (before $Dequeue()$ operations occur) rank $r$ 
being enqueued in it is $\frac{1}{n}$, via coupling we can assume that $e$ is enqueued in the queue
$i$, if the label queue $i$ contains the label with rank $r$. Then, we remove all the extra labels from label queues,
that is if the element with rank $r$ does not exist in the queues, then we remove the label with rank $r$ from the label queues as well. Next, for each $Dequeue()$ operation which chooses queues $i$ and $j$ uniformly at random and proceeds to dequeue from the queue which has the element
with the smaller rank (timestamp) on top, we also check the labels on top of label queues $i$ and $j$ and remove the smaller one.
Notice that this way, at any point in time, if the element with current rank $r$ is in queue $i$, then the the label with current rank $r$ is in label queue $i$ as well, and vice versa. This can be formally proved by induction on $Dequeue()$ operations.
Here, we switch gears and concentrate on proving rank bounds on the process with labels.
The process can be formulated as follows. Let $v_1(t),v_2(t),...,v_n(t)$ be the labels on top of the label queues after $t$ dequeues have occurred. Initially, we have that $v_i(t)=0$, for each $1 \le i \le n$. Then at each step $t+1$, we pick two queues $i$ and $j$
uniformly at random and if w.l.o.g queues $i$ has the smaller label on top, then $v_i(t+1)=v_i(t)+Exp(1/n)$ (for every 
$k \neq i$, we have that $v_k(t+1)=v_k(t)$. Notice the similarity between this process and Algorithm 1.
Let $x_i(t)=\frac{v_i(t)}{n}$. Our initial aim is to upper bound $\Gamma^{con}(t)$ in this case.

\begin{lemma} \label{lem:gammaboundPQ}
Given that $w(t)=\frac{Exp(\frac{1}{n})}{n}$ at every step $t$,
if $\alpha \le \frac{1}{4096 C e^3\log n}$ and $C \ge 433$, then at any time step $t$
	$$\E[\Gamma^{con}(t)] \leq \frac{146 C(\epsilon)n}{\alpha \epsilon}.$$
\end{lemma}

\begin{proof}
Our goal is to apply Lemma \ref{lem:gamma-bound} when $w(t)=\frac{Exp(\frac{1}{n})}{n}$ at every step $t$.
For this we will just need to show that Lemma \ref{lem:goodbound} still holds.
The key steps towards accomplishing this are 
generalizing Lemma \ref{lem:Chernoff} for exponential weights of mean $1$, as opposed to weights of value $1$, since we are no longer able to apply Chernoff's inequality and making sure that (\ref{eqn:needtomodify}) still holds.

First we establish the bounds in (\ref{weightproperty1}) and (\ref{weightproperty2}), in order
to be able to apply lemmas \ref{lem:BadBoundLower} and \ref{lem:BadBoundUpper}.
We know that for each $t$, $w(t)=\frac{Exp(1/n)}{n}$. Clearly $\E[w(t)]=\frac{\E[Exp(1/n)]}{n}=1$.
In this case the moment generating function is $M(x)=\E[e^{x w(t)}]=\E[e^{\frac{x}{n}
Exp(\frac{1}{n})}]=\frac{\frac{1}{n}}{\frac{1}{n}-\frac{x}{n}}=\frac{1}{1-x}$, for $x < 1$. This gives us that $M''(x)=\frac{2}{(1-x)^3}$.
Hence if $\lambda=1$ and $S=8$,  we have that for every so that for every $x<\lambda/2$ we have $M''(x) < 2S$.
This means that to apply Lemma \ref{lem:beta} we will need $\alpha \le \frac{1}{576}$ (which is feasible since we need an upper bound on $\alpha$ to be even smaller in order to prove Lemma $\ref{lem:gamma-bound}$.

Recall that previously we had that $H_{max}(t)=\max\{H_1(t), H_2(t), ..., H_i(t)\}$,
where $H_i(t)$ was the number of times bin $i$ was a random choice made by operations $op_{t-Cn+1}, op_{t-Cn+2}, ..., op_t$
and then we knew that if $op_{t+1}$ was a good operation, total increment received by bin $i$ by the operations
which were concurrent with $op_{t+1}$ was at most $H_i(t)$. In this to have the same property we redefine
$H_i(t)$ as follows.

For $0 \le u \le Cn$, let $z(u)=(z_1(u), z_2(u), ..., z_n(u))$ be the $n$ dimensional vector.
We assume that $z_i(0)=0$ for each $1 \le i \le n$. For $0 \le u < n$, consider operation $op_{t-Cn+u+1}$.
Let $i$ and $j$ be the random bins it chooses, we know that it increments the bin which has the smaller load 
at the time of performed reads by $w(t-Cn+u+1)$. We set $z_i(u+1)=z_i(u)+w(t-Cn+u+1)$, 
$z_j(u+1)=z_j(u)+w(t-Cn+u+1)$ (even if $i=j$) and for $k \neq i,j$ we set $z_k(u+1)=z_k(u)$.

Notice that $z_i(Cn)$ is the upper bound on the total increment received by bin $i$ from operations 
$op_{t-Cn+1}, op_{t-Cn+2}, ..., op_t$. And thus we can set $H_i(t)=z_i(Cn)$. 
Hence, our goal is to upper bound $\max\{z_1(Cn), z_2(Cn), ..., Z_n(Cn)\}$

We use argument similar to Lemma \ref{lem:BadBoundUpper}. 
Let $\Upsilon_i(u)=e^{\frac{z_i(u)}{8}}$ and $\Upsilon(u)=\sum_{i=1}^n \Upsilon_i(u)$ (Thus, $\Upsilon(0)=n$)
If $i$ and $j$ are chosen by operation $op_{t-Cn+u+1}$, then 

\begin{align*}
\E&[\Upsilon(u+1)|z(t)] - \Upsilon(u)\\&=
\E[\Upsilon_i(u+1)|z(t)] - \Upsilon_i(u)]+\E[\Upsilon_j(u+1)|z(t)] - \Upsilon_j(u)]
\\&= \Big(e^{\frac{z_i(u)}{8}}+e^{\frac{z_j(u)}{8}}\Big)
\Bigg(\E\Big[M\Big(\frac{1}{8})\Big)\Big]-1\Bigg) 
\\ &=
\Big(e^{\frac{z_i(u)}{8}}+e^{\frac{z_j(u)}{8}}\Big)
\Bigg(\E\Big[M(0)+\frac{M'(0)}{8}+\frac{M''(\xi)}{2\cdot 8^2}\Big]-1 \Bigg) 
\\&=
\Big(e^{\frac{z_i(u)}{8}}+e^{\frac{z_j(u)}{8}}\Big)
\Bigg(\E\Big[1+\frac{w(t-Cn+u+1)}{8}+\frac{M''(\xi)}{2\cdot 8^2}\Big]-1 \Bigg) 
\\&\le 
\frac{1}{4}\Big(e^{\frac{z_i}{8}}+e^{\frac{z_j}{8}}\Big)
\end{align*}
Where in the the last step we used 
$0 \le \xi \le \frac{1}{8}$, and as we established above $M''(x) \le 2S=16$ for each $x \le \frac{1}{2}$.
Also, recall that the weight is one in expectation at every step.
Hence,
\begin{align*}
\E[\Upsilon(u+1)|z(t)] - \Upsilon(u) &\le \frac{1}{n^2} \sum_{1\le i\le n, 1\le j \le n} \frac{1}{4}
\Big(e^{\frac{z_i}{8}}+e^{\frac{z_j}{8}}\Big)  \\&= \frac{2}{n} \sum_{i=1}^n \frac{1}{4} e^{\frac{z_i}{8}} 
= \frac{\Upsilon(u)}{2n}.
\end{align*}

After removing conditioning we get that 
\begin{align*}
\E[\Upsilon(u+1)] \le (1+\frac{1}{2}) \E[\Upsilon(u+1)].
\end{align*}
After applying the above inequality $Cn$ times we also get that 
\begin{align*}
\E[\Upsilon(Cn)] \le n(1+\frac{1}{2n})^{Cn} \le ne^{C/2}.
\end{align*}
Finally we proceed as in the proof of Lemma \ref{lem:main}
\begin{align*}
Pr[H_{max}(t) > 2KCe^3 \log n] &\le Pr[\exists i : z_i(Cn) > 2KCe^3 \log n] \\&\le Pr[\Upsilon(Cn) > e^{\frac{KCe^3 \log n}{4}}]
\\&\le \frac{ne^{C/2}}{e^{\frac{KCe^3 \log n}{4}}}.
\end{align*}
The last step is to verify that (\ref{eqn:needtomodify}) is still true.
Notice that even though the (\ref{eqn:needtomodify}) uses $\alpha \le \frac{1}{4096 C e^3\log n}$ and $C \ge 2$,
the Lemma \ref{lem:gamma-bound} requires that $C \ge 433$, and we will take advantage of this upper bound:

\begin{align*}
\sum_{K=1}^{\infty}
&\E[\Gamma^{con}(t)|A_K(t)]Pr[A_K(t)] (e^{2\alpha e^3C (K+1)\log{n}}-1)
\\&\overset{\text{Lemma \ref{lem:lowerboundcond}}}{\le} \sum_{K=1}^{\infty}  \frac{\E[\Gamma^{con}(t)]e^{5\alpha e^3C (K+1) \log{n}+\log n+\frac{C}{2}}}
{e^{\frac{K C e^3 \log{n}}{4}}} \\ &\le \sum_{K=1}^{\infty}  \frac{\E[\Gamma^{con}(t)]e^{K+\log n+\frac{C}{2}}}
{e^{2 K C  \log{n}}}.
\end{align*}
We have that $K \ge 1$, $\log n \ge \frac{1}{2}$ (Assuming $n \ge 2$), and $C \ge 433$,
thus we have that $\frac{C}{2} \le KC\log n$, $K \le \frac{KC \log n}{4}$ and $\log n \le \frac{KC \log n}{4}$.
Hence
\begin{align*}
\E&[\Gamma^{con}(t)|A_K(t)]Pr[A_K(t)] (e^{2\alpha e^3C (K+1)\log{n}}-1) \le 
\sum_{K=1}^{\infty}  \frac{\E[\Gamma^{con}(t)]}{e^{\frac{K C  \log{n}}{2}}} \\ &\le \sum_{K=1}^{\infty}  \frac{\E[\Gamma^{con}(t)]}{e^{Kn^{16}}}
\le \frac{\E[\Gamma^{con}(t)]}{2048}.
\end{align*}
With this in place we know that \ref{lem:goodbound} holds if even if $w(t)=\frac{Exp(\frac{1}{n})}{n}$ at every step $t$
and then we can just use Lemma \ref{lem:gamma-bound} to finish the proof.
\end{proof}

Now we are ready to upper bound the ranks of dequeued elements.
\begin{theorem}
\label{thm:queues}
	Assuming an oblivious adversary, the MultiQueue algorithm with parameter $n$ (Algorithm~\ref{algo:approx-queue}) is distributionally linearizable to a sequential randomized relaxed queue $Q_R$, which ensures that at each step $t$, the maximum expected rank of 
	dequeued element is $O(n \log^2 n)$, and average expected rank is $O(n \log n \log \log n)$.
\end{theorem}

\begin{proof}
First we bound the expected maximum rank of the elements on top. 
Recall that $x_{max}(t)$ and $x_{min}(t)$ are the largest and smallest weights of bins after $t$ steps
and let $v_{max}(t)=n x_{max}(t)$ and $v_{max}(t)=n x_{max}(t)$ be the largest and smallest labels on top of queues after $t$ dequeue operations.
We start by showing that for any $1\le i \le n$, 
\begin{align} \label{eqn:labelrankproperty}
\E[rank(v_i(t))] &\le \sum_{1 \le j \le n, j \neq i} (1+\frac{\E\Big[|v_i(t)-v_j(t)|\Big]}{n} \nonumber \\&=\sum_{1 \le j \le n, j \neq i} (1+\E\Big[|x_i(t)-x_j(t)|\Big]).
\end{align}
The proof is similar to the proof of Lemma 11 in \cite{AKLN17}.
We fix (condition on) $v_1(t),v_2(t),...,v_n(t)$.
For each $v_j(t) > v_i(t)$ we know that the label queue $j$ does not contain labels which are smaller than $v_i(t)$,
hence the labels in $j$ do not influence the rank of $v_i(t)$.
In the case when $v_j(t) < v_i(t)$, we know that the number of labels in $j$ which are smaller than $v_i(t)$ is 
one plus the number of labels in $j$ which belong to the interval $(v_j(t), v_i(t))$. We know that the difference 
between consecutive labels in each label queue is $Exp(\frac{1}{n})$ hence the expected number of labels in $j$ which belong to interval $(v_j(t),v_i(t))$
is upper bounded by $\E[Poi(\frac{v_i(t)-v_j(t)}{n})]=\frac{v_i(t)-v_j(t)}{n}$ (This simply follows from the properties of Exponential and Poisson distributions).
Thus, obviously $\E[rank(v_i(t))] \le \sum_{1 \le j \le n, j \neq i} (1+\frac{|v_i(t)-v_j(t)|}{n})$ and (\ref{eqn:labelrankproperty}) follows after removing 
conditioning on $v_1(t),v_2(t),...,v_n(t)$. With this in place we have that for any $1 \le i \le n$
\begin{align*}
\E[rank(v_{i}(t))] &\le \sum_{1 \le j \le n, j \neq i} (1+\E[\Big[|x_i(t)-x_j(t)|\Big]) \\ &\le (n-1)+(n-1)\E[x_{max}(t)-x_{min}(t)]) \\ &= O(n \log^2 n).
\end{align*}
Where the last step follows from the proof of Lemma \ref{lem:main} where it is shown that 
both $\E[x_{max}(t)-\mu(t)]$ and $\E[\mu(t)-x_{min}(t)]$ are $O(\log^2 n)$.

Next we aim to upper bound $\sum_{i=1}^n \frac{\E[rank(v_i(t))]}{n}$.
This is exactly average expected rank of removed label since during removal we choose both label queues uniformly at random.
We have that 
\begin{align} \label{eqn:labelrankproperty_v2}
\sum_{i=1}^n \frac{\E[rank(v_i(t))]}{n} &\le \frac{1}{n} \sum_{i=1}^n \sum_{1 \le j \le n, j \neq i} (1+\E[\Big[|x_i(t)-x_j(t)|\Big])
\nonumber \\ &\le n+\frac{1}{n} \sum_{i=1}^n \sum_{j=1}^n \E[\Big[|x_i(t)-x_j(t)|\Big]
\nonumber \\ &\le n+\frac{1}{n} \sum_{i=1}^n \sum_{j=1}^n \E[\Big[|x_i(t)-\mu(t)|+|x_j(t)-\mu(t)|\Big] 
\nonumber \\&=
n+2\sum_{i=1}^n \E[\Big[|x_i(t)-\mu(t)|\Big].
\end{align}
Using Jensen's inequality we get that 
\begin{align*}
\frac{\sum_{i=1}^n \alpha |x_i(t)-\mu(t)|}{n} &= \log \Bigg(e^{\frac{\sum_{i=1}^n \alpha |x_i(t)-\mu(t)|}{n}}\Bigg) \\&\le 
\log \Bigg(\frac{\sum_{i=1}^n e^{\alpha |x_i(t)-\mu(t)|}}{n}\Bigg) \le \log \Bigg(\frac{\Gamma_{con}(t)}{n}\Bigg).
\end{align*}

Hence
\begin{align*} 
\E\Big[\frac{\sum_{i=1}^n \alpha |x_i(t)-\mu(t)|}{n}\Big] &\le \E\Big[\log \Bigg(\frac{\Gamma_{con}(t)}{n}\Bigg)\Big]
\overset{Jensen}{\le} \log \Bigg(\frac{\E[\Gamma_{con}(t)]}{n}\Bigg) \\&\overset{\text{Lemma \ref{lem:gammaboundPQ}}}{\le}
\log \Bigg(\frac{146 C(\epsilon)}{\alpha \epsilon}\Bigg)=O(\log \log n).
\end{align*}
Where in the last step we used $\alpha=\Theta(\frac{1}{\log n})$ and for the same reason we get that $\sum_{i=1}^n |x_i(t)-\mu(t)|=O(n \log n \log \log n)$.
Finally, (\ref{eqn:labelrankproperty_v2}) gives us that 
\begin{align*}
\sum_{i=1}^n \frac{\E[rank(v_i(t))]}{n} \le n+2O(n \log n \log \log n)=O(n \log n \log \log n).
\end{align*}
\end{proof}

\section{Experimental Results}
\label{sec:tests}

\paragraph{Setup.} 
Our experiments were run on an Intel E7-4830 v3 with 12 cores per socket and 2 hyperthreads (HTs) per core, for a total of 24 threads, and 128GB of RAM.
In all of our experiments, we pinned threads to avoid unnecessary context switches.
Hyperthreading is only used with more than 12 threads.
The machine runs Ubuntu 14.04 LTS.
All code was compiled with the GNU C++ compiler (G++) 6.3.0 with compilation options \texttt{-std=c++11 -mcx16 -O3}.

\paragraph{Synthetic Benchmarks.} 
We implemented and benchmarked the MultiCounter algorithm on a multicore machine. To test the behavior under contention, threads continually increment the counter value using the two-choice process. 
We use no synchronization other than the atomic fetch and increment instruction for the update. 
Figure~\ref{fig:scal} shows the scalability results, while Figure~\ref{fig:qual} shows the ``quality" guarantees of the implementation in terms of values returned by the counter over time, as well as maximum gap between bins over time. Quality is measured in a single-threaded execution, for $64$ counters. (Recording quality accurately in a concurrent execution appears complicated, as it is not clear how to order the concurrent read steps.)

\begin{figure*}[t!]
\centering     
\subfigure[Scalability of the concurrent counter for different values of the ratio $C$ between counters and \# threads.]{\label{fig:scal}\includegraphics[width=0.43\textwidth]{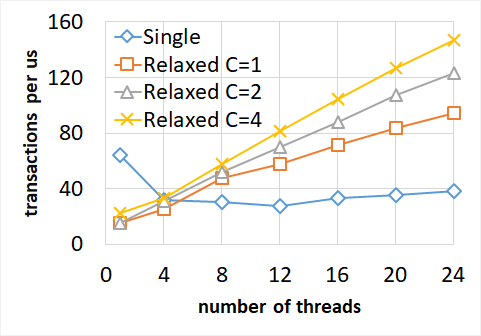}}~~~
\subfigure[Quality results for the concurrent counter in a single-threaded execution. The $x$ axis is \# increments.]{\label{fig:qual}\includegraphics[width=0.43\textwidth]{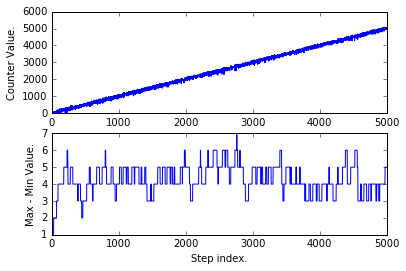}} \\ 
\subfigure[TL2 benchmark, 1M objects.]{\label{fig:scal1}\includegraphics[width=0.32\textwidth]{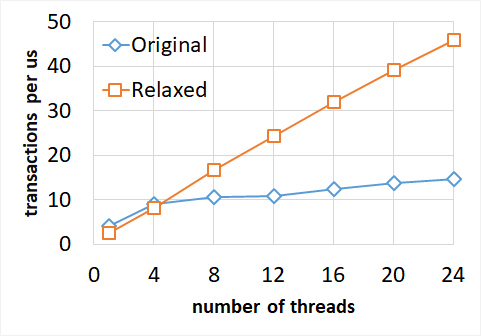}}
\subfigure[TL2 benchmark, 100K objects.]{\label{fig:scal2}\includegraphics[width=0.32\textwidth]{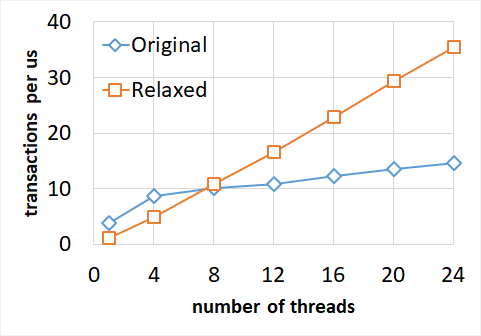}}
\subfigure[TL2 benchmark, 10K objects.]{\label{fig:scal3}\includegraphics[width=0.32\textwidth]{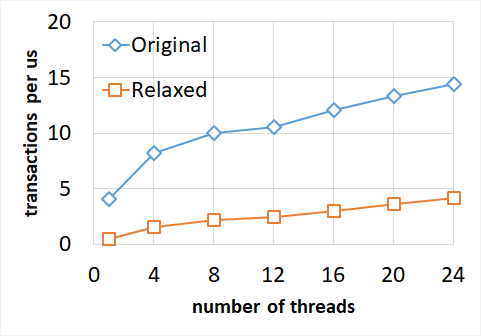}} 
\caption{Experimental Results for the Concurrent Counter}
\end{figure*}

\paragraph{TL2 Benchmark.}
Transactional Locking II (TL2) is a software implementation of transactional memory introduced by~\cite{tl2}.
TL2 guarantees opacity by using fine-grained locking and a global clock $G$. 
TL2 associates a \textit{version lock} with each memory location.
A version lock behaves like a traditional lock, except it additionally stores a version number that represents the value of $G$ when the memory location protected by the lock was last modified.
At a high level, a transaction starts by reading $G$, and uses the clock value it reads to determine whether it ever observes the effects of an uncommitted transaction. 
If so, the transaction will abort.
Otherwise, after performing all of its reads, it locks the addresses in its write set (validating these locations to ensure that they have not been written recently), rereads $G$ to obtain a new version $v'$, performs its writes, then releases its locks, updating their versions to $v'$.

\paragraph{TL2 with Relaxed Global Clocks.} In the standard implementation of TL2, $G$ is incremented using fetch-and-add (FAA).
This quickly becomes a concurrency bottleneck as the number of threads increases, so the the authors developed several improved implementations of $G$. However, they too experience scaling problems at large thread counts. 
We replace this global clock counter $G$ with a MultiCounter implementation, and compare against a highly-optimized baseline implementation.  

Due to the fact that the counter is relaxed, reasoning about the correctness of the resulting algorithm is no longer straightforward. 
In particular, a key property we need to enforce is that the timestamp which a thread writes to a set of objects as part of its transaction (generated when the thread is holding locks to commit and written to all objects in its write set) cannot be held by any other threads at the same time, since such threads might read those concurrent updates concurrently, and believe that they occurred in the past. 
For this reason, we modify the TL2 algorithm so that threads write ``in the future," by adding a quantity $\Delta$, which exceeds the maximum clock skew we expect to encounter in the MultiCounter over an execution, to the maximum timestamp $t_{\max}$ they have encountered during their execution so far. Thus, each new write always increments an object's timestamp by $\geq \Delta$. 
We stress that that the (approximate) global clock is implemented by the MultiCounter algorithm, and that it is disjoint from the object timestamps. 

This protocol induces the following trade-offs. First, the resulting transactional algorithm only ensures safety with high probability, since the $\Delta$ bound might be broken at some point during the execution, and lead to a non-serializable transaction, with extremely low probability. Second, we note that, once an object is written with a timestamp that occurs in the future, transactions which immediately read this object may abort, since they see a timestamp that is larger than theirs. 
Hence, once an object is written, at least $\Delta$ operations should occur \emph{without accessing this object}, so that the system clock is incremented past the read point without causing readers to abort. 
Intuitively, this upper bounds the frequency at which objects should be written to for this approximate timestamping mechanism to be efficient. 
On the positive side, this mechanism allows us to break the scalability bottleneck caused by the global clock. 

We verify this intuition through implementation. See Figures~\ref{fig:scal1}---\ref{fig:scal3}. 
We are given an array of $n$ transactional objects, with $n$ between 10K and 1M. 
Transactions pick 2 array locations uniformly at random, then start a transaction, increment both locations, and then commit the transaction. We record the average throughput out of ten one-second experiments. 
We verify correctness by checking that the array contents are consistent with the number of executed operations at the end of the run; none of these experiments have resulted in erroneous outputs. 
We record the rate at which transactions commit, as a function of the number of threads. 
We note that, for 1M and 100K objects, the average frequency at which each location is written is below the heavy abort threshold, and we obtain almost linear scaling with MultiCounters. 
At 10K objects we surpass this threshold, and see a considerable drop in performance, because of a large number of aborts.

\section{Conclusions and Future Work}

We have presented the first concurrent analysis of the two-choice load-balancing process, showing that this classic randomized algorithm is in fact robust to asynchrony under an oblivious adversary. Our analysis extends existing tools, namely~\cite{PTW15}, in non-trivial ways, in particular by showing that the potential analysis can withstand adversarially corrupted updates. 
Our results have non-trivial practical applications, as they show that a popular set of randomized concurrent data structures in fact provide strong probabilistic guarantees in arbitrary executions, which we express via a new correctness condition called \emph{distributional linearizability}. This inspires a scalable approximate counting mechanism, trading off contention and exactness guarantees, which can be used to scale a transactional application.

%

\bibliography{bibliography}

\end{document}